%% file: arxiv.tex
\documentclass[11pt]{article}
\usepackage[margin=1in]{geometry}

\usepackage{amsmath, amssymb, amsthm}
\usepackage[mathic]{mathtools}

\usepackage{array}
\usepackage{booktabs}
\usepackage{multirow,multicol}

\usepackage{authblk}

\usepackage[T1]{fontenc}
\usepackage{libertine}
\usepackage[libertine]{newtxmath} 
\usepackage[scaled=0.96]{zi4} 

\usepackage{a4wide, dsfont, microtype, xcolor, paralist}

\usepackage{enumitem}

\usepackage[ocgcolorlinks]{hyperref} 
\colorlet{DarkRed}{red!50!black}
\colorlet{DarkGreen}{green!50!black}
\colorlet{DarkBlue}{blue!50!black}
\hypersetup{
	linkcolor = DarkRed,
	citecolor = DarkGreen,
	urlcolor = DarkBlue,
	bookmarksnumbered = true,
	linktocpage = true
}

\usepackage{orcidlink}  

\input{./arxiv_macros.tex}

\newcommand{\oracle}{\texttt{SM}}

\makeatletter
\newcommand\notsotiny{\@setfontsize\notsotiny\@vipt\@viipt}
\makeatother

\begin{document}

\author{Riccardo Maso}

\author{Nicola Prezza}

\author{Carlo Tosoni}

\affil{DAIS, Ca' Foscari University of Venice, Italy\\

\small{ \{ \texttt{riccardo.maso}, \texttt{nicola.prezza}, \texttt{carlo.tosoni}\}\texttt{\@unive.it}}}

\title{Faster Cache-Efficient Pattern Matching for Deterministic Wheeler Pangenome Graphs}
\date{}

\maketitle
\begin{abstract}
  Pattern matching on strings is regarded as one of the core operations in computer science. Although researchers proposed several solutions to this problem, some of the most elegant and widely used approaches are based on the renowned Burrows-Wheeler transform (BWT). The success of the BWT lies in its pattern matching algorithm known as backward search, which is not only near-optimal in the RAM model, but also runs directly on a compressed representation of the input string. More recently, the backward search has been generalized to Wheeler deterministic finite automata (DFAs), a subclass of standard DFAs, without losing its near-optimal time efficiency. Similarly to the case of strings, this pattern matching algorithm for Wheeler DFAs has found applications in bioinformatics, where researchers have shown that specific pangenome graphs of human chromosomes can be transformed into Wheeler DFAs and consequently indexed using this strategy. However, this BWT-based index on Wheeler DFAs inherited a significant drawback from the original backward search, namely the high number of I/O operations triggered during the algorithm execution, which are in the worst-case lower-bounded by the length of the pattern. In this paper, we address this limitation by proposing the first cache-friendly algorithm specifically designed for Wheeler DFAs. Our new data structure reduces the number of I/O operations by employing a strategy analogous to the suffix array: it interleaves a binary search with fast sequential scans of the automaton. We empirically validate this new indexing strategy by running our algorithm on real-world Wheeler pangenome graphs. We show that while our data structure can use up to 15 times the space required by the backward search, it can also be $500$ times faster and able to process a single character of the pattern in less than $3$ ns.
\end{abstract}

\begin{acknowledgements}
  We would like to thank Daniel Puttini for providing us the pangenome graphs employed in the experimental section of the paper, and Davide Cenzato for  fruitful discussions on the implementation details.
  All authors are funded by the European Union (ERC, REGINDEX, 101039208). Views and opinions expressed are however those of the authors only and do not necessarily reflect those of the European Union or the European Research Council Executive Agency. Neither the European Union nor the granting authority can be held responsible for them.
\end{acknowledgements}


\section{Introduction}
Pattern matching, the process of identifying a pattern $P$ within a text $T$, is a fundamental operation in computer science. 
This task finds applications in several domains, with bioinformatics being a notable example, where locating the occurrences of a pattern can be used to perform sequence alignment or map short sequencing reads to a large reference genome~\cite{bowtie, surveySequenceAlignment}.
Due to these reasons, this problem has been deeply studied in the literature where various solutions have been proposed.
Some of these solutions pre-process the pattern $P$ to achieve search times proportional to the length of the text $T$~\cite{KnuthMP77}, others reverse this paradigm by preprocessing the text $T$ and their query times depend on the length of the pattern $P$~\cite{suffix-tree-Weiner,suffix-tree-McCreight,suffix-array}.
This latter group includes the pattern matching algorithms based on the Burrows-Wheeler transform (BWT)~\cite{burrows1994block}, including the renowned FM-index~\cite{FM-index}.
The importance of the FM-index lies in its dual capability to provide efficient query time, while compressing the text into its empirical entropy.
This \emph{pattern matching} problem naturally generalizes to a (potentially infinite) collection of strings, where the objective is to identify whether a pattern $P$ occurs as a substring within this collection.
A common approach to represent this collection is to employ a finite-state automaton, accepting the collection of strings as its regular language.
However, contrary to the case of strings, pattern matching on general automata cannot be solved in strongly subquadratic time unless the Orthogonal Vectors Hypothesis~\cite{OVH} is false, as proven by Equi et al.~\cite{EquiMT21,EquiMTG23}.
In order to overcome this complexity constraint, research has focused on identifying subclasses of automata that allow for efficient pattern matching. 
This led to the definition of \emph{Wheeler NFAs}~\cite{Wheeler-graphs}, which represents a natural generalization of the original BWT to automata.
Informally, an NFA is said to be Wheeler if there exists a total order of its states \emph{consistent} with the strings entering into them.
This total ordering represents a natural generalization of the permutation returned by the original string BWT, which sorts the string characters based on their subsequent suffixes.
This line of research has found applications in bioinformatics, where genomic collections may be represented by using a pangenome graph, encoding the genetic diversity of a population~\cite{pangenome-cons}.
Indeed, Sirén et al.~\cite{Jouni-indexing-graphs} showed that some pangenome graphs built from a collection of human chromosomes can be transformed into Wheeler NFAs for supporting queries across the genetic variation of a population.
The current state-of-the-art algorithm for indexing Wheeler NFAs relies on the \emph{forward search}, a generalization of the well-known backward search used by the string BWT, but processing the pattern in a left-to-right fashion.
While in the RAM model this forward search runs in near-optimal $\widetilde O(\lvert P \rvert)$ time, in the I/O model it inherited a notable bottleneck from the backward search, namely the high number of cache misses, which in the worst-case are lower-bounded by the length of the pattern.

In the string domain, the suffix array and tree support cache-friendly solutions for this problem.
Indeed, the occurrences of an input pattern can be located by means of a binary search on the suffix array, combined with fast sequential scans of the matching regions in the text.
However, these solutions have not yet been generalized to Wheeler NFAs.
This represents a significant drawback for the usage of Wheeler NFAs in real-world applications, since in the era of the so-called big data applications time efficiency has become as paramount as space complexity for indexing large-scale datasets.

\subsection{Our contributions}

In this paper, we address this limitation by proposing the first cache-friendly algorithm specifically designed for \emph{deterministic} Wheeler automata.
We term our new indexing strategy the \emph{Graph Suffix Array} as it represents a natural generalization of the aforementioned pattern matching algorithm based on the suffix array of a string.
Indeed, our algorithm locates the occurrences of a pattern by alternating a binary search with fast linear scans of the input Wheeler DFA, thus achieving a drop in the overall number of cache misses triggered during the algorithm execution.
Specifically, our algorithm works in three distinct phases.
During the first phase, we identify the longest prefix of the input pattern $P$ occurring at least twice in the Wheeler automaton.
We show that this operation can be performed efficiently in the I/O model by running a binary search on the list of \emph{infimum} and \emph{supremum} strings~\cite{match-stat, suff-double}, which are respectively the smallest and largest strings entering into each state of the Wheeler DFA.
These infimum and supremum strings can be encoded using a pseudoforest, i.e., a forest where every tree node has at most one incoming edge, which ensures that every tree contains at most one cycle.
As a consequence of this, it follows that we can employ a specialized version of the \emph{heavy-light decomposition} to achieve a cache-efficient traversal of these pseudotrees.
Due to these observations, we show that the following result holds.

\begin{lemma}\label{lemma: longest prefix}
    Let $\mathcal D$ be a Wheeler DFA of $n$ states, and let $P$ be an input pattern of length $m$.
    Then we can compute the longest prefix $\beta$ of $P$ reaching at least two distinct states in $\mathcal D$ using $\mathcal O( (\frac{m}{B}+\log n)\log n \log m )$ operations in the I/O model, where $B$ is the block size.
\end{lemma}

Therefore, if this prefix $\beta$ is equal to $P$, then we return the states reached by $\beta$ and we are done.
Otherwise, we know that the prefix $\beta a$ of $P$ can reach at most a single state of the DFA.
Thus, in the second phase of the algorithm execution we use standard algorithms and data structures on Wheeler automata~\cite{Wheeler-graphs} to check in $\mathcal O(\log \log \sigma)$ I/O operations if there exists a transition labeled $a$ outgoing from these states reached by $\beta$.
If this transition exists, then we conclude that the string $\beta a$ reaches a state $u$ of the DFA.
Therefore, in this last phase of the algorithm we start navigating the DFA from $u$, by following the remaining string suffix $\gamma$, where $P = \beta a \gamma$.
In order to speed up this last phase, we partition the Wheeler DFA in \emph{unary paths}, that is, maximal sequences of connected states having in-degree equal to $1$.
By doing so, we bound the number of cache misses in this last phase by $\mathcal O(\frac{m}{B} + d)$, where the parameter $d$ indicates the number of unary paths traversed while following $\gamma$ in the DFA.
As a consequence of this, by summing up the time complexities of these three phases, we obtain the following result.

\begin{theorem}\label{thm: main graph suffix array}
    Let $\mathcal D$ be a Wheeler DFA of $n$ states, and let $P$ be an input pattern of length $m$.
    Then there exists a data structure taking $\mathcal O(\lvert \mathcal D \rvert)$ RAM words and able to locate the occurrences of $P$ in $\mathcal D$ in at most $\mathcal O( (\frac{m}{B}+ \log n)\log n \log m + d)$ operations in the I/O model, with $B$ being the block size.
\end{theorem}

In this theorem, the value $\lvert \mathcal D \rvert$ denotes the number of states and transitions of $\mathcal D$.
Specifically, the space complexity of $\mathcal O(\lvert \mathcal D \rvert )$ stems from the third phase of the algorithm, which requires a plain representation of the input automaton.
As the parameter $d$ is crucial for the time efficiency of the Graph Suffix Array, we carry out experiments on real-world Wheeler deterministic pangenomes, built over human chromosomes, to experimentally assess the typical values of $d$ in these graphs. 
We empirically show that, due to the high sparseness of these pangenome graphs~\cite{pang-unary}, this parameter $d$ is about two orders of magnitude smaller than the pattern length.
We also run experiments showing that our algorithm is up to $500$ times faster compared to the original pattern matching algorithm for Wheeler automata based on the aforementioned forward search.
As a limitation of our approach, the space usage of $\mathcal O(\lvert \mathcal D \rvert)$ RAM words is asymptotically higher than the $O(\lvert \mathcal D\rvert \log \sigma )$ bits required by the forward search~\cite[Lemma 4]{Wheeler-graphs}, where $\sigma$ is the alphabet size.
This spatial analysis is also empirically validated, as we show that our data structure can use up to $15$ times the space required by the forward search.
However, in our experiments the space usage of our data structures never exceeds $1.5$ GB, which implies that our solution can be easily stored in the RAM.

Finally, as a direct consequence of Lemma~\ref{lemma: longest prefix}, we show that if we aim to locate only those patterns occurring at least twice in the automaton, then we can achieve a tighter time and space complexity.

\begin{corollary}\label{cor: GrSA for patterns occurring twice}
    Let $\mathcal D$ be a Wheeler DFA of $n$ states, and let $P$ be an input pattern of length $m$.
    Then there exists a data structure taking $\mathcal O(n)$ RAM words which returns the set $T$ of states reached by $P$, if $\lvert T \rvert \geq 2$, and the emptyset $\emptyset$ otherwise.
    The overall number of operations in the I/O model is bounded by $\mathcal O( (\frac{m}{B}+ \log n)\log n)$, where $B$ is the block size.
\end{corollary}

The rest of the paper is organized as follows.
In Section~\ref{sec: preliminaries}, we introduce the notation and the preliminary concepts used in the paper.
In Section~\ref{sec:algorithm}, we prove some properties on Wheeler DFAs and we introduce our Graph Suffix Array.
In Section~\ref{sec:data-strucure}, we prove the time and space complexity of our novel pattern matching algorithms.
Finally, in Section~\ref{sec: experiments} we conduct experiments on real pangenome graphs using our new indexing strategy.

\section{Notation and preliminaries}\label{sec: preliminaries}

\paragraph{Relations and strings.} A relation $R$ over a set $U$ is a subset of the set $U \times U$.
In this paper, we work with \emph{total orders}, which are relations satisfying (i) reflexivity, (ii) antisymmetry, (iii) transitivity, and (iv) totality.
We use the symbols $\leq$ and $\preceq$ to denote total orders.
Given an alphabet $\Sigma$, we denote by $\Sigma^*$ the set of all finite strings over $\Sigma$, where $\varepsilon \in \Sigma^*$ denotes the empty string. 
In addition, we define $\Sigma^{\omega}$ as the set of all strings formed by an infinite enumerable concatenation of characters from $\Sigma$. 
In particular, we work with \emph{left-infinite} strings, meaning that $\alpha \in \Sigma^{\omega}$ is constructed from $\varepsilon$ by prepending an infinite number of characters to it. 
For a string $\alpha \in \Sigma^*$, we write $\alpha[i,j]$ to denote the substring consisting of characters from position $i$ through $j$, where $1 \leq i \leq j \leq \lvert \alpha \rvert$.
For two strings $\alpha\in\Sigma^*$ and $ \beta\in \Sigma^*\cup\Sigma^{\omega}$, we write $\alpha \dashv\beta$ to denote that $\alpha$ is a \emph{suffix} of $\beta$. 
Moreover, we assume to have a fixed total order $\preceq$ over the alphabet $\Sigma$, and we extend this order $\preceq$ to $\Sigma^*\cup\Sigma^{\omega}$ \emph{co-lexicographically} as follows. For every two strings $\alpha, \beta \in \Sigma^* \cup \Sigma^{\omega}$ we have that (i) $\varepsilon \preceq \alpha$ for every $\alpha \in \Sigma^* \cup \Sigma^{\omega}$, and
(ii) if $\alpha = \alpha' a$ and $\beta = \beta' b$ for some $\alpha', \beta' \in \Sigma^* \cup \Sigma^{\omega}$ and $a, b \in \Sigma$, then $\alpha \prec \beta$ if and only if $(a \prec b) \vee (a = b \wedge \alpha' \prec \beta')$.

\paragraph{DFA.} A deterministic finite automaton (DFA) is a $4$-tuple $\mathcal D = (Q,\Sigma,\delta, s)$ defined as follows.
$Q$ is the set of states, $\Sigma$ is a finite, nonempty set called the alphabet.
$\delta$ is a function $\delta: Q \times \Sigma \rightarrow Q\cup\{\bot\}$ mapping each pair $(u,a)$, with $u\in Q$ and $a\in \Sigma$, to either a state in $Q$ or to a special symbol $\bot$, indicating that no transition labeled by $a$ leaves the state $u$. 
Finally, $s \in Q$ is the unique initial state of the automaton.
Typically, a DFA is defined by a 5-tuple including the set of final states $F \subseteq Q$, which we omit since we are not interested in distinguishing between final and non-final states. 
Given a DFA $\mathcal{D} = (Q,\Sigma,\delta,s)$, a state $u \in Q$, and a character $a \in \Sigma$, we may use the shortcut $\delta_a(u)$ for $\delta(u,a)$.
Moreover, we extend the transition function to words $\alpha \in \Sigma^*$ as follows: let $a\in \Sigma, \alpha\in\Sigma^*$, and $u\in Q$, then $\delta(u,\alpha a)=\delta(\delta(u,\alpha),a)$ and $\delta(u,\varepsilon)=u$. 
We make the following assumptions on the DFAs we treat: (i) We assume the alphabet $\Sigma$ to be effective; meaning that every character labels at least one transition in the automaton.
(ii) We assume that every state is reachable from the initial state. 
(iii) The initial state $s$ has a single incoming transition $s = \delta(s, \#)$, where $\# \in \Sigma$ is a special symbol of the alphabet not labeling any other transition, and satisfying $\# \preceq a$ for all $a \in \Sigma$. 
We note that these assumptions are not restrictive, as any DFA can be transformed to satisfy them without changing its regular language. 
In the rest of the paper, we denote by $n$ the number of states in the input DFA, formally $n = \lvert Q \rvert$. 

\paragraph{Infimum and supremum strings.} Let $\mathcal D$ be a DFA, and let $u\in Q$.
We define $I_u$ as the set of all left-infinite strings obtained by following transitions backward starting from $u$, and prepending the corresponding character at every step.
Formally, we define the set $I_u$ as follows.
\begin{definition}
    Let $\mathcal{D}=(Q,\Sigma,\delta,s)$ be a DFA and $u\in Q$. We define $I_u$ as:
    \begin{align*}
        I_u=&\{\alpha\in \Sigma^{\omega}\mid \text{there exist an infinite sequence of states } u_1,u_2,\dots \text{ in } Q \text{ such that}\\
        & (i)\  u_1=u, (ii)\ u_i =\delta_a(u_{i+1}),\text{ for every }i\geq 1 \text{ where } a = \alpha_i)\}       
    \end{align*}

\end{definition}

In this definition we consider $\alpha_i$ to be the $i$-th character of $\alpha \in \Sigma^\omega$ from right to left, since $\Sigma^\omega$ is the set of \emph{left-infinite} strings.
We now define the infimum and supremum strings of a state $u \in Q$, as introduced by Conte et al.~\cite{match-stat} and Kim et al.~\cite{suff-double}, representing the lower and upper bounds of the strings in $I_u$, respectively.
Our definition aligns with that of Becker et al.~\cite{encoding-linear}.

\begin{definition}[Infimum and supremum strings]\label{def:inf-sup}
Let $\mathcal{D}=(Q,\Sigma,\delta,s)$ be a DFA, and let $u$ be a node in $Q$. 
Then, the infimum string of $u$, denoted by $\inf I_u$, is the co-lexicographically smallest string in $I_u$. 
The supremum string of $u$, denoted by $\sup I_u$, is the co-lexicographically largest string in $I_u$.
\end{definition}
See Figure~\ref{fig: example} for an example.
This definition is equivalent to the formulation originally introduced by Conte et al.~\cite{match-stat} and Kim et al.~\cite{suff-double}, with the distinction that in Definition~\ref{def:inf-sup} the infimum and supremum strings are always infinite strings.
The existence of $\inf I_u$ and $\sup I_u$ for each $u \in Q$ is guaranteed by Observation 8 of the article of Kim et al.~\cite{suff-double}.

\paragraph{Wheeler DFA.} In the following we report the definition of Wheeler DFAs, which was first introduced by Gagie et al.~\cite{Wheeler-graphs}. 
\begin{definition}[Wheeler DFA]
  Let $\mathcal{D}=(Q,\Sigma,\delta,s)$ be a DFA. A Wheeler order on $\mathcal{D}$ is a total order $\leq$ on $Q$ such that for every $v = \delta(u,a)$ and $v'=\delta(u',a')$ the following two axioms hold:
  \begin{enumerate}[label=\emph{(Axiom~\arabic*)}, leftmargin=6em, itemsep=0.01em]
    \item\label{axiom1} If $a \prec a'$, then $v<v'$.
    \item If $a=a'$ and $u < u'$, then $v<v'$.
  \end{enumerate}
  A DFA that admits a Wheeler order is called a \emph{Wheeler DFA} (WDFA).
\end{definition}

\ref{axiom1} implies a property known in the literature as \emph{input-consistency}: in a Wheeler DFA, all transitions entering a state $u$ must share the same label $a \in \Sigma$. Indeed, if two transitions entering $u$ had distinct labels, this would imply $u < u$, a contradiction.
Generally, Gibney and Thankachan~\cite{hardness-wheeler} showed that recognizing a Wheeler automaton is an NP-complete problem~\cite{hardness-wheeler}.
However, Alanko et al.~\cite[Theorem~3.3]{regular-prefix} proved that in the deterministic setting this problem can be solved in linear time, with respect to the DFA size.
Furthermore, in the same paper they also showed that a Wheeler DFA admits a unique Wheeler order $\leq$ which can be computed in linear time with respect to the automaton dimension as well.
Henceforth, given a Wheeler DFA $\mathcal D$ we denote by $u_1, u_2, \ldots, u_n$ the sequence of states of $\mathcal D$ sorted according to its unique Wheeler order $\leq$, i.e., $u_i \leq u_j$ holds if and only if $i \leq j$.
In the following, we recall that WDFAs satisfy a fundamental structural property, known in the literature as \emph{path coherence}.
This property ensures that the states reached by any string form a contiguous interval in the Wheeler state order. 

\begin{lemma}[Path-coherence]\cite[Lemma 3]{Wheeler-graphs}\label{lemma:path-coherence}
     Let $\mathcal{D}$ be a Wheeler DFA and $\leq$ its Wheeler order.
     Then for every string $\alpha \in \Sigma^*$ there exists an interval $[i,j]$, with $1 \leq i,j \leq n$, such that $u_i, u_{i+1}, \dots, u_j$ are those and only states of $\mathcal D$ reached by the string $\alpha$.
     If no state of $\mathcal D$ can be reached by $\alpha$, then $i > j$ holds.
\end{lemma}

As a consequence of path-coherence, it follows that a Wheeler order permutes the states on an automaton consistently with the strings reaching them.
Path-coherence also ensures that Wheeler NFAs can be indexed using a strategy analogous to other BWT-based data structures~\cite{Wheeler-graphs}, which is known in the literature as \emph{forward search}.
The next lemma from the article of Conte et al.~\cite{match-stat} establishes a connection between infimum/supremum strings and Wheeler orders in DFAs.
\begin{lemma}\cite[Lemma 3]{match-stat}\label{lemma:iff-sup-inf}
    Let $\mathcal{D}=(Q,\Sigma,\delta,s)$ be a WDFA and let $\leq$ be its Wheeler order. For any distinct states $u, v \in Q$ with $u < v$, it holds that $\sup I_u \preceq \inf I_v$.
\end{lemma}

Finally, for an input DFA we define $T(\alpha)$ as the set of states reached by a walk whose transitions are labeled by the string $\alpha \in \Sigma^*$.
This set can be formally defined using the sets $I_u$ as follows: $T(\alpha) = \{u\in Q \mid (\exists \beta \in I_u)(\alpha \dashv \beta)\}$.
We recall that by Lemma~\ref{lemma:path-coherence} in Wheeler DFAs this set $T(\alpha)$ forms an interval according to its Wheeler order $\leq$.

\paragraph{Computational model.}
In this work, we analyze the time complexity of our algorithms in the I/O model. 
Specifically, we assume a two-level memory hierarchy, which is formed by the RAM and the cache.
The RAM is formed by contiguous blocks of size $B$ and the performance of our algorithms is strictly evaluated by the number of block transfers from the RAM to the cache.
For the analysis of the space complexity, we assume to have RAM words of size $\Omega(\log \lvert \mathcal D \rvert)$, where $\lvert \mathcal D \rvert$ is the input dimension.

\section{The Graph Suffix Array}\label{sec:algorithm}
In this section, we formally introduce our Graph Suffix Array for solving \texttt{count} and \texttt{locate} queries over a WDFA.
Formally, for every string $\alpha\in \Sigma^*$ the query \texttt{count($\alpha$)} returns the cardinality of the set $T(\alpha)$, while \texttt{locate($\alpha$)} returns the set $T(\alpha)$ itself.
We refer at this two queries as \emph{subpath queries}. 
As the set $T(\alpha)$ plays a central role for both queries, in this section we propose an efficient algorithm to compute it.
We start this section by proving some preliminary properties on this set $T(\alpha)$, which are at the core of the indexing strategy proposed in this article.
\begin{remark}\label{lemma:decreasing-matching}
    Let $\mathcal{D}$ be a DFA, $\alpha\in\Sigma^*$ and $a \in \Sigma$, then  $|T(\alpha a)| \leq |T(\alpha)|$.
\end{remark}
\begin{proof}
By definition, every state in $ |T(\alpha a)| $ is a successor of some state in $ |T(\alpha)| $ via a transition labeled $ a $. Because the transition function $ \delta $ is deterministic, each state has at most one successor for every symbol $a \in \Sigma$. Hence, the number of such successors cannot exceed the number of starting states, that is: $|T(\alpha a)| \leq |T(\alpha)|$.
\end{proof}
In the paper introducing Wheeler automata, Gagie et al.~\cite{Wheeler-graphs} presented a strategy for indexing the automaton based on the famous forward search, which characterizes BWT-based data structures.
We propose a new approach based on binary search, which takes inspiration from other renowned text indexing techniques~\cite{CSA}.
Our key idea is to derive $T(\alpha)$ through an intermediate set $\widetilde{T}(\alpha)$, whose purpose is to approximate it.
Formally, for a given DFA this approximated set is defined as $\widetilde{T}(\alpha)=\{u\in Q \mid \alpha \dashv \inf I_u \vee \alpha \dashv  \sup I_u \}$.
We now prove that in Wheeler DFAs, the set $\widetilde{T}(\alpha)$ exhibits a strong relationship with $T(\alpha)$. 

\begin{lemma}\label{lemma:T-tilde-relation}
    Let $\mathcal D$ be a Wheeler DFA and $\leq$ its Wheeler order.
    Then for every string $\alpha \in \Sigma^{*}$ the following properties hold:
    \begin{enumerate}[label=(\arabic*), leftmargin=3em, itemsep=0.01em]
    \item\label{lemma:T-tilde-1} If $\widetilde{T}(\alpha)\neq\emptyset$ then $\widetilde{T}(\alpha)=T(\alpha)$,
    \item\label{lemma:T-tilde-2} If $\widetilde{T}(\alpha)=\emptyset$ then $\lvert T(\alpha)\rvert\leq 1$.
  \end{enumerate}
  \end{lemma}

\begin{proof}

\emph{\ref{lemma:T-tilde-1}}
By definition we know that for every $u \in \widetilde{T}(\alpha)$ it holds that $\alpha \dashv \inf I_u \vee \alpha \dashv \sup I_u$.
Since by Definition~\ref{def:inf-sup} we know $\inf I_u, \sup I_u \in I_u$, it follows that $ u \in T(\alpha)$.
This proves $\widetilde{T}(\alpha) \subseteq T(\alpha)$.
Now let us assume that $\widetilde{T}(\alpha) \neq \emptyset$.
To complete the proof we have to show that $T(\alpha) \subseteq \widetilde{T}(\alpha)$. 
Let $v$ and $w$ be respectively the smallest and largest states in $\widetilde{T}(\alpha)$ according to $\leq$, and consider the infinite strings $\beta_v \in \{ \inf I_v, \sup I_v \}$ and $\beta_w \in \{ \inf I_w, \sup I_w \}$ satisfying $\alpha \dashv \beta_v, \beta_w$.
Suppose by contradiction that there exists a state $u\in T(\alpha) \setminus \widetilde{T}(\alpha)$.
If $v < u < w$, by Lemma~\ref{lemma:iff-sup-inf} and by Definition~\ref{def:inf-sup} we know that $\beta_v \preceq \inf I_u \preceq \sup I_u \preceq \beta_w$. However, $\alpha \dashv \beta_v \land  \alpha \dashv \beta_w $ implies $\alpha \dashv \inf I_u \land  \alpha \dashv \sup I_u $, a contradiction.
Therefore, let us now consider the case $u<v$, the remaining case $w<u$ is analogous. 
By Lemma~\ref{lemma:iff-sup-inf}, we know that $\sup I_u \preceq \beta_v$.
Moreover, since $u \in T(\alpha)$, there exists $\beta \in I_u$ such that $\alpha\dashv\beta$, and by definition of supremum string we have $\beta \preceq \sup I_u \preceq \beta_v$.
However, $\alpha \dashv \beta$ and $\alpha \dashv \beta_v$ imply $\alpha \dashv \sup I_u$, a contradiction.

\emph{\ref{lemma:T-tilde-2}} We show the contrapositive, i.e., that $\lvert T(\alpha)\rvert > 1$ implies $\widetilde{T}(\alpha) \neq \emptyset$.
Let $u,v\in T(\alpha)$ be two states such that $u < v$, it follows that there exist $\beta_u\in I_u$ and $\beta_v\in I_v$ satisfying $\alpha\dashv \beta_u,\beta_v$.
Moreover, Definition~\ref{def:inf-sup} and Lemma~\ref{lemma:iff-sup-inf} imply that $\beta_u\preceq \sup I_u\preceq\inf I_v \preceq \beta_v$.
This in turn implies that $\alpha \dashv\sup I_u$ and $\alpha \dashv\inf I_v$, since $\alpha$ is a proper suffix of both strings $\beta_u$ and $\beta_v$.
Therefore, we conclude that $\widetilde{T}(\alpha) \neq \emptyset$.
\end{proof}

We observe that this result implies also that the set $\widetilde{T}(\alpha)$, like $T(\alpha)$, forms a state interval with respect to the Wheeler order, since $\widetilde{T}(\alpha)=T(\alpha)$ or $\widetilde{T}(\alpha)=\emptyset$ hold.
In the following corollary we show that the cardinality of the set $\widetilde{T}(\alpha)$ is non-increasing.

\begin{corollary}\label{cor:tilde-dec-matching}
    Let $\mathcal{D}$ be DFA, $\alpha\in \Sigma^*$ and $a\in \Sigma$, then $\lvert \widetilde{T}(\alpha a)\rvert \leq \lvert \widetilde{T}(\alpha)\rvert$.
\end{corollary}

\begin{proof}
    By combining Lemma~\ref{lemma:T-tilde-relation} and Remark~\ref{lemma:decreasing-matching}, it follows that the only case that needs to be proved is that if $\lvert \widetilde T(\alpha a) \rvert = 1$, then $\lvert \widetilde T(\alpha) \rvert \geq 1$.
    Therefore, consider $\widetilde T(\alpha a) = \{u\}$, and suppose $\alpha a \dashv \sup I_u$, the other case $\alpha a \dashv \inf I_u$ is analogous.
    By the maximality of $\sup I_u$, it follows that there exists a state $v$, such that $\alpha \dashv \sup I_v$, and consequently it holds that $\widetilde T(\alpha) \neq \emptyset$.
\end{proof}

\subsection{The pattern matching algorithm}

We now present our novel indexing technique, which supports both \texttt{count} and \texttt{locate} queries for any string $\alpha \in \Sigma^*$.
Recall that, as previously discussed, computing the set $T(\alpha)$ suffices to answer both queries.
We now begin by introducing a preliminary notion that will be used extensively in this algorithm.
We recall that in the following $u_1, u_2, \ldots, u_n$ represents the ordered list of state with respect to the Wheeler order $\leq$.

\begin{definition}
    Let $\mathcal{D}=(Q,\Sigma,\delta,s)$ be a WDFA and let $\leq$ be its Wheeler order. We define the list $\mathcal{K}_{\mathcal D}$ as:
        \[
        \mathcal{K}_{\mathcal D} = \bigl( \inf I_{u_1}, \sup I_{u_1}, \inf I_{u_2}, \sup I_{u_2}, \dots, \inf I_{u_n}, \sup I_{u_n} \bigr)
        \]       
\end{definition}
We emphasize that, by Lemma~\ref{lemma:iff-sup-inf}, the strings in $\mathcal{K_D}$ are sorted co-lexicographically.
Consequently, the elements of $\mathcal{K_D}$ that are suffixed by the same string $\alpha$ constitute an interval in $\mathcal{K_D}$, which implies that also $\widetilde{T}(\alpha)$ forms an interval in $\mathcal{K_D}$.
The algorithm assumes access to an oracle that we term \emph{suffix match} and denote by \oracle.
Given a string $\alpha \in \Sigma^*$, the oracle $\oracle(\alpha)$ returns the maximal interval $]i, j[$, with $0\leq i < j \leq 2n+1$, such that: (i)~either $i=0$ or $\mathcal{K}_{\mathcal{D}}[i] \prec \alpha$, (ii)~either $j=2n+1$ or $\alpha \prec \mathcal{K_D}[j]$, and (iii)~$\alpha \dashv \mathcal{K}_{\mathcal{D}}[k]$ for every $k$ such that $i < k < j$. 

Note that, by definition of $\widetilde{T}(\alpha)$, once we obtain the interval $]i,j[$ output by $\oracle(\alpha)$ we can easily reconstruct the sequence $u_{i'}, u_{i'+1}, \ldots , u_{j'}$ corresponding to $\widetilde{T}(\alpha)$ as $i' = \lceil (i + 1)/2 \rceil$ and $j' = \lceil (j-1)/2 \rceil$.
Our algorithm consists of two main parts.
In the first part, we start by finding the longest prefix $\beta$ of $\alpha$ such that $\widetilde{T}(\beta) \neq \emptyset$.
This operation is done by performing a binary search over the length of the input string $\alpha$, where at every iteration we call our oracle $\oracle$ to compute the corresponding set $\widetilde{T}(\alpha)$.
Note that, we can find such string through a binary search, since by Corollary~\ref{cor:tilde-dec-matching} the cardinality of $\widetilde{T}$ is non-increasing when we consider longer prefixes of $\alpha$.
From Property~\ref{lemma:T-tilde-1} of Lemma~\ref{lemma:T-tilde-relation}, we know that since $\widetilde{T}(\beta)$ is not empty we have that $\widetilde{T}(\beta) = T(\beta)$.
Consequently, if $\alpha = \beta$, then we can directly use $\widetilde{T}(\beta)$ to compute \texttt{count}($\alpha$) and \texttt{locate}($\alpha$) and we are done.
Otherwise, let $a \in \Sigma$ and $\gamma \in \Sigma^*$ be respectively the character and the string satisfying $\alpha = \beta a \gamma$.
By definition of $\beta$, we have that $\widetilde{T}(\beta a)=\emptyset$, and it follows from Property~\ref{lemma:T-tilde-2} of Lemma~\ref{lemma:T-tilde-relation} that $\lvert T(\beta a)\rvert \leq 1$. 
Consequently, in this case we identify the unique state $u_j$ that may belong to $T(\beta a)$.
We will show that such a state $u_j$ can be identified through a call to our oracle $\oracle(\beta a)$.
Next, we check if there exists a transition labeled $a$, leaving a state $v \in \widetilde{T}(\beta)$.
In particular, we will observe that $T(\beta a) = \{u_j \}$ holds if and only if such a transition exists, which must necessarily enter into $u_j$.
In this case, to complete the algorithm execution it is sufficient to understand if there exists an integer $k$, with $1 \leq k \leq n$, such that $u_k = \delta(u_j, \gamma)$.
If these conditions are not satisfied we will conclude that $\texttt{count}(\alpha) = 0$ and $\texttt{locate}(\alpha) = \emptyset$.

\paragraph{Algorithm execution.}
We now present a formal description of the algorithm, a pseudocode formulation is presented in Algorithm~\ref{alg:computeT}.
The inputs of our algorithm are the input string $\alpha \in \Sigma^*$, $\mathcal{D}= (Q, \Sigma, \delta, s)$ the Wheeler DFA  and the oracle $\oracle$.
The algorithm starts by performing a binary search over the length of $\alpha$ to find the longest prefix $\beta$ of $\alpha$ such that $\widetilde{T}(\beta)$ is not empty.
To do so, we initialize the endpoints of the binary search $[hi,lo]$ to $[1,\lvert \alpha \rvert]$, and set the prefix $\beta$ to the empty string $\varepsilon$.
At each iteration of the binary search, we define $\beta'$ as the prefix of $\alpha$ of length $mid=lo + \lfloor (hi-lo)/2 \rfloor$, and we compute the open interval $]i,j[$ by calling the oracle $\oracle(\beta')$.
If $j > i+1$, the interval contains at least one value, implying that $\widetilde{T}(\beta') \neq \emptyset$ due to the oracle definition.
Therefore, in this case, we update the endpoint $lo$ to $mid+1$ to search for a longer prefix satisfying the desired property, and we update $\beta$ to $\beta'$ since $\lvert \beta \rvert < \lvert \beta'\rvert$.
Otherwise, if the interval is empty (i.e., $j=i+1$), then $\widetilde{T}(\beta')$ is also empty. It follows from Corollary~\ref{cor:tilde-dec-matching} that any prefix $\beta''$ longer than $\beta'$ must necessarily satisfy the condition $\widetilde{T}(\beta'') = \emptyset$.
Therefore, in this second case we update the endpoint $hi$ to $mid-1$, to search if a shorter prefix satisfies the desired property.
As a consequence of that, upon termination, $\beta$ is guaranteed to be the longest prefix of $\alpha$ such that $\widetilde{T}(\beta)\neq\emptyset$. 
Once $\beta$ is found, we compute the open interval $]i, j[$ returned by the oracle $\oracle(\beta)$.
We then determine the indices relative to the Wheeler order $\leq$ by updating the values to $i = \lceil(i+1)/2\rceil$ and $j = \lceil(j-1)/2\rceil$, such that $\widetilde{T}(\beta) = \{u_i, u_{i+1}, \dots, u_j\}$.
If $\alpha = \beta$, we return this set, since $T(\alpha)=\widetilde{T}(\beta)$.
Otherwise, let $a$ be the character succeeding $\beta$ in $\alpha$.
We query the oracle $\oracle(\beta a)$ to obtain its corresponding interval $]i', j'[$.
Since $\widetilde{T}(\beta a) = \emptyset$, it follows that $i'+1 = j'$, since $i'+1 < j$ would imply the existence of a state $u$ such that either $\beta a \dashv \inf I_u$ or $\beta a \dashv \sup I_u$.
We then update the indices to $i' = \lceil i'/2\rceil$ and $j' = \lceil j'/2\rceil$.
If $i' + 1 = j'$, it follows that $\beta a$ does not reach any state in the Wheeler DFA, since in this case $\sup I_{u_{i'}} \prec \beta a \prec \inf I_{u_{j'}}$ and therefore $\beta a$ cannot suffix any string in $I_u$, for every $u \in Q$.
Therefore, in this case we conclude that $T(\beta a )$ is empty and thus we return $\emptyset$.
Otherwise, if $i' = j'$, then we know that for the state $u_{i'}$ it holds that $\inf I_{u_{i'}} \prec \beta a \prec \sup I_{u_{i'}}$.
Therefore, in this case $u_{i'}$ is the only candidate state that can belong to $T(\beta a)$, thus we want to check if there exists a string $\alpha' \in I_{u_{i'}}$ such that $\beta a \dashv \alpha'$.
Consider, the sequence $u_i, u_{i+1}, \ldots, u_{j}$ forming $T(\beta)$.
To do so, we check if there exists $k \in [i, j]$ such that $\delta(u_k, a) \neq\; \perp$.
This verifies whether there exists a transition labeled $a$ leaving any state in $T(\beta)$.
If such an integer $k$ does not exist, then since $u_i, u_{i+1}, \ldots, u_{j}$ are those and only states reached by $\beta$, we conclude that $T(\beta a)$ is empty and so we return $\emptyset$.
On the other hand, if such an integer $k$ exists, then the state reached by this transition $\delta(u_k,a)$ belongs to $T(\beta a)$ by construction, and therefore since $u_{i'}$ is the only candidate that can be part of $T(\beta a)$, in this case we conclude that $\delta(u_k, a) = u_{i'}$.
Therefore, to terminate the algorithm we initialize $\gamma$, such that $\alpha = \beta a \gamma$, and we compute $w = \delta(u_{i'}, \gamma)$.
Since $\alpha = \beta a \gamma$, it follows that in this latter case that if $w =\; \perp$, then $T(\alpha) = \emptyset$ and so we return the empty set.
Otherwise, if $w \in Q$, due to an analogous reasoning, we return the set $T(\alpha) = \{ w\}$ and we are done.

\paragraph{Computing $\mathbf{\widetilde{T}}$.}

In this section, we have seen that the binary search shown in Line~\ref{alg:computeT function locate states} of Algorithm~\ref{alg:computeT} returns the longest prefix $\beta$ of $\alpha$ for which $\widetilde{T}(\beta) \neq \emptyset$ holds.
Therefore, by combining this result with Lemma~\ref{lemma:T-tilde-relation} and Corollary~\ref{cor:tilde-dec-matching}, it follows that we can also compute $T(\alpha)$, if $\lvert T(\alpha) \rvert \geq 2$, by just calling the oracle $\oracle$ once over the string $\alpha$ without having to navigate the input DFA.
To achieve this, once we have called $\oracle(\alpha)$, we return the set $\widetilde T(\alpha)$ if $\lvert \widetilde T(\alpha) \rvert \geq 2$, and $\emptyset$ otherwise.
See Algorithm~\ref{alg:computeT-tilde} for a pseudocode. 

\begin{algorithm}[!ht]

\caption{\texttt{ComputeT}($\alpha$)}\label{alg:computeT}
\Input{A WDFA $\mathcal D = (Q,\Sigma,\delta,s)$, a string $\alpha \in \Sigma^*$, the oracle $\oracle$}
\Output{The set $T(\alpha)$}

\medskip

\Fn{\LocateStates{$\alpha, \oracle$}}{\label{alg:computeT function locate states}

$lo \gets 1,\;hi \gets \lvert \alpha \rvert$\tcp*{Extremes of the binary search}
$\beta \gets \varepsilon$\tcp*[r]{Initialise $\beta$ as the empty string}

\While(\tcp*[f]{Compute longest prefix $\beta$ with $\widetilde{T}(\beta) \neq \emptyset$}){$lo \leq hi$}{

    \smallskip

    $mid \gets lo+\lfloor (hi-lo)/2 \rfloor$\;
    $\beta' \gets \alpha[1,mid]$\tcp*[r]{Candidate to update $\beta$}
    $i,j\gets \oracle(\beta')$\;
    
    \If(\tcp*[f]{If $\widetilde{T}(\beta') \neq \emptyset$}){$j > i+1$}{
        $lo \gets mid+1$, $\,\beta \gets \beta'$ \tcp*[r]{update lowest extreme and $\beta$}
    }\lElse(\tcp*[f]{Otherwise update the highest extreme}){$hi\gets mid-1$}
}

$i,j\gets \oracle(\beta)$\;
$i\gets \lceil (i+1)/2 \rceil,\ j\gets \lceil (j-1)/2 \rceil$\tcp*[r]{$T(\beta) = \{u_i, u_{i+1}, \ldots , u_j\}$}

\smallskip

\Return $i,j, \beta$\;\label{alg:computeT function locate states end}
}

\medskip

$i,j,\beta \gets $\LocateStates{$\alpha, \oracle$}\;

\smallskip

\lIf(\tcp*[f]{If $T(\alpha) = T(\beta)$}){$\beta=\alpha$}{\Return $\{u_{i}, u_{i+1}, \ldots, u_{j}\}$}

\smallskip

$a \gets \alpha[\lvert \beta\rvert+1]$\;
$i',j' \gets \oracle(\beta a)$\;
$i'\gets \lceil i'/2 \rceil,\ j'\gets \lceil j'/2 \rceil$\tcp*[r]{Either $T(\beta a) = \emptyset $ or $T(\beta a) = \{u_{i'}\}$}

\medskip

\lIf(\tcp*[f]{Case 1) $T(\beta a) = \emptyset$}){$i' < j'$}{\Return $\emptyset$}

\smallskip

\If(\tcp*[f]{Case 2) $T(\beta a) = \{u_{i'} \}$}){$\exists k \in [i,j]$ s.t. $\delta(u_k,a) \neq\; \perp$}{\label{alg: alg:computeT second step}

    \smallskip

        $\gamma \gets \alpha[\lvert\beta\rvert + 2,  \lvert\alpha\rvert ]$\tcp*{$\alpha = \beta a \gamma$}\label{alg:computeT navigating start}
        $w\gets \delta(u_{i'},\gamma)$ \tcp*[r]{compute state $w$ reached by $\alpha$}

        \smallskip
        
        \lIf{$w\neq\; \perp$}{\Return $\{w\}$}\label{alg:computeT navigating end}
}

\Return $\emptyset$
\end{algorithm}

\begin{algorithm}[!ht]

\caption{\texttt{Compute$\widetilde{T}$}($\alpha$)}\label{alg:computeT-tilde}
\Input{A WDFA $\mathcal D = (Q,\Sigma,\delta,s)$, a string $\alpha \in \Sigma^*$, and the oracle $\oracle$}
\Output{The set $T(\alpha)$, if $\lvert T(\alpha) \rvert \geq 2$ and $\emptyset$ otherwise}

\smallskip

$i,j\gets \oracle(\alpha)$\;

\smallskip

$i \gets \lceil (i+1) / 2 \rceil, j \gets \lceil (j-1) / 2\rceil$\;

\smallskip

\If(\tcp*[f]{If $\lvert T(\alpha)\rvert \geq 2$}){$i < j$}{
    \Return $\{u_i, u_{i+1}, \ldots, u_j \}$\tcp*[r]{Return $T(\alpha) = \{u_i, u_{i+1}, \ldots , u_j\}$}
}
\Return $\emptyset$
    
\end{algorithm}

\section{Time and space complexity}\label{sec:data-strucure}

In Section~\ref{sec:algorithm}, we presented our new indexing strategy and we showed its correctness.
Thus, to complete the proofs of Lemma~\ref{lemma: longest prefix}, Theorem~\ref{thm: main graph suffix array}, and Corollary~\ref{cor: GrSA for patterns occurring twice} we now analyze the space and time complexities of our novel pattern matching algorithms.
Specifically, we partition this analysis into three different parts, each corresponding to one of the three substrings into which the pattern is split (i.e., $\alpha = \beta a \gamma$).

\paragraph{Part 1, the Oracle.}
We present an implementation of the oracle $\oracle$ which, given a string $\beta' \in \Sigma^*$, returns the endpoints of the maximal interval $]i,j[$ defined in Section~\ref{sec:algorithm}.
Our approach requires storing the \emph{infimum and supremum pseudoforests} of the automaton $\mathcal D$ as defined by Kim et al.~\cite{suff-double}.
In each of these pseudoforests, every node has exactly one incoming edge; following the unique infinite walk entering a node yields its corresponding infimum (or supremum) string.
Figure~\ref{fig: example} provides a graphical representation of these pseudoforests.
At every call of the oracle $\oracle(\beta')$, for a string $\beta' \in \Sigma^*$, we perform a binary search over the interval $[1, 2n]$ to identify the greatest value $i$ such that $\mathcal{K_D}[i] \prec \beta'$, if this value does not exist we set $i=0$. 
At each step of the binary search, we need to co-lexicographically compare $\beta'$ with an entry in $\mathcal{K_D}$, that is either an infimum or a supremum string depending on the index.
Each comparison requires reconstructing the corresponding infimum (supremum) string by traversing the infimum (supremum) pseudoforest in a backward fashion, i.e., towards the tree root, starting from the corresponding node.
Once we have retrieved this integer $i$, we perform a second binary search in the interval $[i+1, 2n]$ to identify the smallest value $j$ such that $\beta' \prec \mathcal{K_D}[j] \wedge \neg (\beta' \dashv \mathcal{K_D}[j])$ by using the same strategy.

To traverse these pseudoforests in a cache-friendly manner, we extend the standard Heavy-Light Decomposition (HLD)~\cite{HLD} from rooted trees to pseudotrees, and consequently, to pseudoforests.
In HLD, every internal node $u$ has a unique \emph{heavy} edge $(u,v)$, where $v$ is the child of $u$ that has the largest subtree.
The heavy paths, formed by grouping heavy edges together, are stored in contiguous memory locations.
All the other edges from $u$ to its children are designated as \emph{light} edges.
This ensures that any node-to-root path encounters at most $\mathcal O( \log n)$ light edges.
To extend this logic, we partition a pseudotree into its unique cycle and a set of disjoint subtrees. 
The cycle acts as a circular root; each node in the cycle serves as the root of the subtree containing all the nodes reachable via paths that do not traverse cycle edges. 
We linearize the cycle and apply HLD to each subtree independently, where the root of each subtree has a pointer to its position in the linearized cycle. 
Since following heavy edges cause $\mathcal O(\frac{m}{B})$ cache misses, and we can encounter at most $\mathcal O( \log n)$ light edges, we deduce that the total number of cache misses for this operation is $\mathcal O(\frac{m}{B} + \log n)$.
Therefore, the overall complexity of a single call to the oracle can cause at most $\mathcal O((\frac{m}{B} + \log n) \log n)$ cache misses, since $\oracle$ executes this operation for every iteration of the binary search over the interval $[1,2n]$.
This completes the proof of Corollary~\ref{cor: GrSA for patterns occurring twice}, because in Algorithm~\ref{alg:computeT-tilde} we call $\oracle$ only once and both HLD require $\mathcal O(n)$ RAM words of space to be stored (every pseudoforest has $n$ nodes and $n$ edges).
Consequently, since Algorithm~\ref{alg:computeT} calls the oracle $\oracle$ at most $\mathcal O(\log m)$ times, the total number of cache misses from Line~\ref{alg:computeT function locate states} to Line~\ref{alg:computeT function locate states end} is bounded by $\mathcal O((\frac{m}{B} + \log n)\log n \log m)$, which consequently completes the proof of Lemma~\ref{lemma: longest prefix}.

\paragraph{Part 2, the outgoing transition.} 
In the second phase of Algorithm~\ref{alg:computeT}, corresponding to Line~\ref{alg: alg:computeT second step}, we have to determine whether there exists a transition labeled $a \in \Sigma$ leaving a state within the range $[u_i, u_j]$, where $\{u_i, u_{i+1}, \dots, u_j\} = T(\beta)$.
To this end, we utilize the Wheeler NFA representation proposed by Gagie et al.~\cite{Wheeler-graphs}, which stores outgoing transitions via two sequences: $\mathsf{OUT}$, containing edge labels ordered by the Wheeler rank of their source states, and $\mathsf{OUT\_DEG}$, a bitvector encoding the out-degrees in unary (see Figure~\ref{fig: example}). 
By augmenting the bitvector $\mathsf{OUT\_DEG}$ with auxiliary data structures, we can support constant-time rank~\cite{rank-const} and select~\cite{select-const} operations in $n+\lvert\delta\rvert + o(n+\lvert\delta\rvert)$ bits.
Moreover, by representing $\mathsf{OUT}$ with alphabet partitioning~\cite[Theorem 3.2]{alphabet-part}, we can encode the sequence using $\mathcal{O} (\lvert\delta\rvert \log \sigma)$ bits, while supporting rank operations in $\mathcal{O} (\log \log \sigma)$.
Since we are assuming that the alphabet is effective, and since Wheeler DFAs are input-consistent, we conclude that $\sigma \leq n$.
Consequently, the I/O complexity of this operation is dominated by the complexity of the previous part.
As shown by Gagie et al.~\cite{Wheeler-graphs}, these operations are sufficient to verify the existence of a transition labeled $a$ from the set $T(\beta)$.

\paragraph{Part 3, navigating the automaton.} 
The third component, which is used to check the final segment $\gamma$, is a representation of the original DFA. 
In this third part, to decrease the number of cache misses, we aim to exploit the \emph{sparseness} of the input automaton $\mathcal D$.
Indeed, in applications like bioinformatics, researchers showed that pangenome graphs may be particularly sparse and linear, which means that most of the nodes have just one outgoing edge~\cite{pang-unary}.
Therefore, to achieve a cache-friendly solution, we employ a \emph{path-compression} over the input DFA.
Path compression collapses sequences of states having exactly one outgoing edge into a contiguous location of the memory.
We call each of these collapsed sequences a \emph{unary path}, whose traversal can be executed in a cache-friendly manner through a sequential scan.
Formally, we store contiguously in memory the maximal sequences of states $\hat u_1, \hat u_2, \ldots , \hat u_p$ satisfying: (i)~the states $\hat u_1, \hat u_2, \ldots \hat u_{p-1}$ have out-degree equal to $1$, (ii)~the states $\hat u_2, \hat u_3, \ldots \hat u_{p}$ have in-degree equal to $1$, and (iii)~for every $i \in [p-1]$, it holds that $\hat u_{i+1} \in \delta_a(\hat u_i)$, for some $a \in \Sigma$.
Figure~\ref{fig: example} shows all the maximal unary paths for a given DFA $\mathcal D$.

To navigate between unary paths, we use the following data structure.
Let $v$ be the last state of an arbitrary unary path, and let $out_v$ be its out-degree.
We use a minimal perfect hash function $h_v : \Sigma \rightarrow [out_v]$ and a hash table $H$ of size $out_v$, storing at position $h_v(a)$ the pointer to the state $z$ such that $z \in \delta_a(v)$ (if any).
By using the minimal perfect hash function described in the book of Ferragina~\cite[Theorem 8.8]{PearlsOfAlgorithms}, we can evaluate the functions $h_v$ in $O(1)$, therefore we have $\mathcal O(d)$ cache misses in the worst-case required to move from a unary path to another, where $d$ is the number of unary paths we encounter while traversing $\mathcal D$ using the string $\gamma$.
Moreover, since each unary path is stored in contiguous positions of the memory, traversing the unary paths can trigger at most $\mathcal O (\frac{\lvert \gamma\rvert }{B})$ cache misses, which becomes $\mathcal O(\frac{m}{B})$ since $\lvert \gamma \rvert < m$.
Therefore, this data structure allows us to execute Algorithm~\ref{alg:computeT} from Line~\ref{alg:computeT navigating start} to Line~\ref{alg:computeT navigating end} triggering at most $\mathcal O(d + \frac{m}{B})$ cache misses.
Despite the fact that the parameter $d$ is $\Theta(m)$ in the worst-case on arbitrary DFAs, in Section~\ref{sec: experiments} we experimentally observe that on pangenome graphs, due to their linear structure, is significantly smaller.

Concerning the space complexity, the minimal perfect hash function $h_v$ described by Ferragina uses $\mathcal O(\lvert out_v \rvert)$ RAM words~\cite[Theorem 8.8]{PearlsOfAlgorithms}.
Thus, the total space complexity becomes $\mathcal O(\lvert \mathcal D \rvert)$ RAM words, where $\lvert \mathcal D \rvert$ is the input dimension (number of transitions plus number of states), when we account for all the states, as well as the other components of $\mathcal D$.
Therefore, by accounting the three parts, we obtain an overall number of cache misses given by $\mathcal O((\frac{m}{B} + \log n)\log n \log m + d)$, while the space complexity is $\mathcal O(\lvert \mathcal D \rvert)$ RAM words, which completes the proof of Theorem~\ref{thm: main graph suffix array}.

\begin{figure}

\begin{center}
\resizebox{0.9\textwidth}{!}{
\begin{tikzpicture}[
    circledim/.style={minimum size=1.6em, font={\scriptsize}, line width = 0.65pt, inner sep=0.15em},
    rectnodes/.style={shape=rectangle, inner sep=0.3em, draw=black, line width = 0.15pt},
    values/.style={font={\footnotesize}, anchor = center, inner sep=0.12em},
    label/.style={font={\footnotesize}, anchor = south, inner sep=0.10em, minimum height = 1em},
    ghost/.style={text centered, minimum size=0, inner sep=0},
    >={Stealth[length=1.6mm, width=1.2mm]},
    scale = 0.80,
    node distance = 0.05
]
    \tikzmath{
        \xdist = 1.7; 
        \ydist = 0.9; 
        \ydistrect = 0.8; 
        \xdistrect = 1.7; 
        \xinrectdist = 0.6; 
    }

    \node[ghost] (.) at (-1.25*\xdist,0) {$(a)$};

    \node[state,circledim] (0) at (0,1*\ydist) {1};
    \node[state,circledim] (1) at (0,0) {2};
    \node[state,circledim] (10) at (1*\xdist,0) {11};
    \node[state,circledim] (14) at (2*\xdist,0) {15};
    \node[state,circledim] (8) at (3*\xdist,0) {9};
    \node[state,circledim] (3) at (4*\xdist,0) {4};
    \node[state,circledim] (11) at (5*\xdist,0) {12};
    \node[state,circledim] (7) at (6*\xdist,0) {8};
    \node[state,circledim] (2) at (7*\xdist,0) {3};
    \node[state,circledim] (6) at (8*\xdist,0) {7};

    \node[state,circledim] (15) at (3*\xdist,1*\ydist) {16};
    \node[state,circledim] (9) at (4*\xdist,1*\ydist) {10};

    \node[state,circledim] (13) at (7*\xdist,1*\ydist) {14};
    \node[state,circledim] (4) at (8*\xdist,1*\ydist) {5};
    \node[state,circledim] (12) at (9*\xdist,1*\ydist) {13};

    \node[state,circledim] (5) at (3*\xdist,-1*\ydist) {6};

    \draw[->, loop right] (0) to node[right] {\footnotesize$\#$} (0);

    \draw[->, red, bend right = 70] (0.west) to node[left] {$a$} (1.west);
    \draw[->, darkgreen] (1) to node[above] {$c$} (10);
    \draw[->, darkgreen] (10) to node[above] {$c$} (14);
    \draw[->, blue] (14) to node[above] {$b$} (8);
    \draw[->, red] (8) to node[above] {$a$} (3);
    \draw[->, darkgreen] (3) to node[above] {$c$} (11);
    \draw[->, blue] (11) to node[above] {$b$} (7);
    \draw[->, red] (7) to node[above] {$a$} (2);
    \draw[->, blue] (2) to node[above] {$b$} (6);

    \draw[->, darkgreen, bend left = 20] (14) to node[above] {$c$} (15);
    \draw[->, darkgreen, bend left = 20] (7) to node[above] {$c$} (13);
    \draw[->, red, bend right = 15] (14) to node[above] {$a$} (5);

    \draw[->, red] (5) to node[above] {$a$} (1*\xdist,-1*\ydist) to [bend left = 21] (1);

    \draw[->, blue, bend left = 29] (6) to (7*\xdist,-1.15*\ydist) to [bend left = 29] (7);
    \node[blue] (!) at (7*\xdist,-1.15*\ydist+0.3) {\small $b$};

    \draw[->, blue] (5) to node[above] {$b$} (5*\xdist,-1*\ydist) to [bend right = 25] (7);

    \draw[->, blue] (9) to node[above] {$b$} (5*\xdist,1*\ydist) to [bend left = 25] (7);

    \draw[->, blue] (15) to node[above] {$b$} (9);

    \draw[->, red, bend left = 35] (6) to node[below] {$a$} (2);

    \draw[->, red] (13) to node[above] {$a$} (4);
    \draw[->, darkgreen] (4) to node[above] {$c$} (12);

    \draw[->, darkgreen] (12.north) to [bend right = 30](8.7*\xdist,1.7*\ydist) to node[below] {$c$}(3*\xdist,1.7*\ydist) to [bend right = 37] (14);

\begin{scope}[yshift=-2.4cm]
    \node[ghost] (.) at (-1.25*\xdist,0) {$(b)$};

    \node[ghost] (0) at (4.15*\xdist,0) {
        \setlength{\tabcolsep}{0.275em}
        \begin{tabular}{m{4em} m{0.4em} m{0.15em} m{0.15em} m{0.15em}m{0.15em} m{0.15em} m{0.15em} m{0.15em} m{0.15em} m{0.15em} m{0.15em} m{0.15em} m{0.15em}m{0.15em} m{0.15em} m{0.15em} m{0.15em} m{0.15em} m{0.15em} m{0.15em} m{0.15em} m{0.15em}m{0.15em} m{0.15em} m{0.15em} m{0.15em} m{0.15em} m{0.15em} m{0.15em} m{0.15em} m{0.15em} m{0.15em} m{0.15em} m{0.15em} m{0.15em} m{0.15em} m{0.15em} m{0.15em} m{0.15em} m{0.15em} m{0.15em}m{0.15em} m{0.15em} m{0.15em} m{0.15em} m{0.15em} m{0.15em} m{0.15em} m{0.15em} m{0.15em} m{0.15em} m{0.15em} }

        {\footnotesize $\mathsf{OUT\_DEG}$} & $=$ & $0$ & $0$ & $1$ & $0$ & $1$ & $0$ & $1$ & $0$ & $1$ & $0$ & $1$ & $0$ & $0$ & $1$ & $0$ & $0$ & $1$ & $0$ & $0$ & $1$ & $0$ & $1$ & $0$ & $1$ & $0$ & $1$ & $0$ & $1$ & $0$ & $1$ & $0$ & $1$ & $0$ & $0$ & $0$ & $1$ & $0$ & $1$ \\ [3pt]
        {\footnotesize $\mathsf{OUT}$} & $=$ & {\notsotiny$\#$} & $a$ & & $c$ & & $b$ & & $c$ & & $c$ & & $a$ & $b$ & & $a$ & $b$ & & $a$ & $c$ & & $a$ & & $b$ & & $c$ & & $b$ & & $c$ & & $a$ & & $a$ & $b$ & $c$ & & $b$ & \\
\end{tabular}
    };
\end{scope}

\begin{scope}[yshift=-4.9cm]
    \node[ghost] (.) at (-1.25*\xdist,0) {$(c)$};

    \node[state,circledim] (0) at (0,1*\ydist) {1};
    \node[state,circledim] (1) at (0,0) {2};
    \node[state,circledim] (10) at (1*\xdist,0) {11};
    \node[state,circledim] (14) at (2*\xdist,0) {15};
    \node[state,circledim] (8) at (3*\xdist,0) {9};
    \node[state,circledim] (3) at (4*\xdist,0) {4};
    \node[state,circledim] (11) at (5*\xdist,0) {12};
    \node[state,circledim] (7) at (6*\xdist,0) {8};
    \node[state,circledim] (2) at (7*\xdist,0) {3};
    \node[state,circledim] (6) at (8*\xdist,0) {7};

    \node[state,circledim] (15) at (3*\xdist,1*\ydist) {16};
    \node[state,circledim] (9) at (4*\xdist,1*\ydist) {10};

    \node[state,circledim] (13) at (7*\xdist,1*\ydist) {14};
    \node[state,circledim] (4) at (8*\xdist,1*\ydist) {5};
    \node[state,circledim] (12) at (9*\xdist,1*\ydist) {13};

    \node[state,circledim] (5) at (3*\xdist,-1*\ydist) {6};

    \draw[->, loop right] (0) to node[right] {\footnotesize$\#$} (0);

    \draw[->, red, bend right = 70] (0.west) to node[left] {$a$} (1.west);
    \draw[->, darkgreen] (1) to node[above] {$c$} (10);
    \draw[->, darkgreen] (10) to node[above] {$c$} (14);
    \draw[->, blue] (14) to node[above] {$b$} (8);
    \draw[->, red] (8) to node[above] {$a$} (3);
    \draw[->, darkgreen] (3) to node[above] {$c$} (11);
    \draw[->, blue] (2) to node[above] {$b$} (6);

    \draw[->, darkgreen, bend left = 20] (14) to node[above] {$c$} (15);
    \draw[->, darkgreen, bend left = 20] (7) to node[above] {$c$} (13);
    \draw[->, red, bend right = 15] (14) to node[above] {$a$} (5);

    \draw[->, blue] (5) to node[above] {$b$} (5*\xdist,-1*\ydist) to [bend right = 25] (7);

    \draw[->, blue] (15) to node[above] {$b$} (9);

    \draw[->, red, bend left = 35] (6) to node[below] {$a$} (2);

    \draw[->, red] (13) to node[above] {$a$} (4);
    \draw[->, darkgreen] (4) to node[above] {$c$} (12);

\end{scope}
\begin{scope}[yshift=-8cm]

    \node[ghost] (.) at (-1.25*\xdist,0) {$(d)$};

    \node[state,circledim] (0) at (0,1*\ydist) {1};
    \node[state,circledim] (1) at (0,0) {2};
    \node[state,circledim] (10) at (1*\xdist,0) {11};
    \node[state,circledim] (14) at (2*\xdist,0) {15};
    \node[state,circledim] (8) at (3*\xdist,0) {9};
    \node[state,circledim] (3) at (4*\xdist,0) {4};
    \node[state,circledim] (11) at (5*\xdist,0) {12};
    \node[state,circledim] (7) at (6*\xdist,0) {8};
    \node[state,circledim] (2) at (7*\xdist,0) {3};
    \node[state,circledim] (6) at (8*\xdist,0) {7};

    \node[state,circledim] (15) at (3*\xdist,1*\ydist) {16};
    \node[state,circledim] (9) at (4*\xdist,1*\ydist) {10};

    \node[state,circledim] (13) at (7*\xdist,1*\ydist) {14};
    \node[state,circledim] (4) at (8*\xdist,1*\ydist) {5};
    \node[state,circledim] (12) at (9*\xdist,1*\ydist) {13};

    \node[state,circledim] (5) at (3*\xdist,-1*\ydist) {6};

    \draw[->, loop right] (0) to node[right] {\footnotesize$\#$} (0);

    \draw[->, darkgreen] (1) to node[above] {$c$} (10);
    \draw[->, blue] (14) to node[above] {$b$} (8);
    \draw[->, red] (8) to node[above] {$a$} (3);
    \draw[->, darkgreen] (3) to node[above] {$c$} (11);
    \draw[->, blue] (11) to node[above] {$b$} (7);
    \draw[->, red] (7) to node[above] {$a$} (2);
    \draw[->, blue] (2) to node[above] {$b$} (6);

    \draw[->, darkgreen, bend left = 20] (14) to node[above] {$c$} (15);
    \draw[->, darkgreen, bend left = 20] (7) to node[above] {$c$} (13);
    \draw[->, red, bend right = 15] (14) to node[above] {$a$} (5);

    \draw[->, red] (5) to node[above] {$a$} (1*\xdist,-1*\ydist) to [bend left = 21] (1);

    \draw[->, blue] (15) to node[above] {$b$} (9);

    \draw[->, red] (13) to node[above] {$a$} (4);
    \draw[->, darkgreen] (4) to node[above] {$c$} (12);

    \draw[->, darkgreen] (12.north) to [bend right = 30](8.7*\xdist,1.7*\ydist) to node[below] {$c$}(3*\xdist,1.7*\ydist) to [bend right = 37] (14);
    
\end{scope}
\begin{scope}[yshift=-11.5cm]

    \node[ghost] (.) at (-1.25*\xdist,0) {$(e)$};

    \node[values] (1) at (0,0) {2};
    \node[label, red, below=-0.06em of 1.south] (1label) {$a$};
    \node[values, right=\xinrectdist of 1] (10) {11};
    \node[label, darkgreen, below=-0.06em of 10.south] (10label) {$c$};

    \node[rectnodes, fit={(1) (1label) (10) (10label)}] (1rect) {};

    \node[values, above=\ydistrect of 1] (0) {1};
    \node[label, below=-0.06em of 0.south] (0label) {\scriptsize $\#$};
    \node[rectnodes, fit={(0) (0label)}] (0rect) {};
    \draw[->, loop right] (0rect) to node[right] {\footnotesize$\#$} (0rect);

    \draw[->, red] (0rect.south) to node[right] {$a$} (0rect|-1rect.north);

    \node[values, right=\xdistrect of 10] (14) {15};
    \node[label, anchor = south, darkgreen, below=-0.06em of 14] (14label) {$c$};
    \node[rectnodes, fit={(14) (14label)}] (14rect) {};
    \draw[->, darkgreen] (1rect) to node[above] {$c$} (14rect);

    \node[values, right=\xdistrect of 14] (8) {9};
    \node[label, blue, below=-0.06em of 8] (8label) {$b$};
    \node[values, right=\xinrectdist of 8] (3) {4};
    \node[label, red, below=-0.06em of 3] (3label) {$a$};
    \node[values, right=\xinrectdist of 3] (11) {12};
    \node[label, darkgreen, below=-0.06em of 11] (11label) {$c$};

    \node[rectnodes, fit={(8) (8label) (3) (3label) (11) (11label)}] (8rect) {};
    \draw[->, blue] (14rect) to node[above] {$b$} (8rect);

    \node[values, above=\ydistrect of 8] (15) {16};
    \node[label, darkgreen, below=-0.06em of 15.south] (15label) {$c$};
    \node[values, right=\xinrectdist of 15] (9) {10};
    \node[label, blue, below=-0.06em of 9.south] (9label) {$b$};
    
    \node[rectnodes, fit={(15) (15label) (9) (9label)}] (15rect) {};

    \draw[->, darkgreen, bend left = 25] (14rect.north) to node[above] {$c$} (15rect.west);

    \node[values, below=1*\ydistrect of 8] (5) {6};
    \node[label, red, below=-0.06em of 5.south] (5label) {$a$};
    \node[rectnodes, fit={(5) (5label)}] (5rect) {};
    \draw[->, red, bend right = 20] (14rect.south) to node[above] {$a$} ($(5rect.west)+(0,0.2)$);
    
    \draw[->, red] ($(5rect.west)+(0,-0.2)$) to ++(-\xdistrect,0) to [bend left = 15] node[above] {$a$}(1rect.south);
    
    \node[values, right=\xdistrect of 11] (7) {8};
    \node[label, blue, below=-0.06em of 7] (7label) {$b$};
    \node[rectnodes, fit={(7) (7label)}] (7rect) {};

    \draw[->, blue] (8rect.east) to node[above] {$b$} (7rect.west);
    \draw[->, blue, bend left = 15] (15rect.east) to node[above] {$b$} (7rect.north);
    \draw[->, blue] (5rect.east) to ++(\xdistrect,0) to [bend right = 20] node[above] {$b$} (7rect.south);

    \node[values, right=\xdistrect of 7] (2) {3};
    \node[label, red, below=-0.06em of 2] (2label) {$a$};
    \node[values, right=\xinrectdist of 2] (6) {7};
    \node[label, blue, below=-0.06em of 6] (6label) {$b$};
    \node[rectnodes, fit={(2) (2label) (6) (6label)}] (2rect) {};

    \draw[->, red] (7rect.east) to node[above] {$a$} (2rect.west);
    \draw[->, red, loop below] (2rect) to node[below] {$a$} (2rect);

    \node[values, above=\ydistrect of 2] (13) {14};
    \node[label, darkgreen, below=-0.06em of 13] (13label) {$c$};
    \node[values, right=\xinrectdist of 13] (4) {5};
    \node[label, red, below=-0.06em of 4] (4label) {$a$};
    \node[values, right=\xinrectdist of 4] (12) {13};
    \node[label, darkgreen, below=-0.06em of 12] (12label) {$c$};
    \node[rectnodes, fit={(13) (13label) (4) (4label) (12) (12label)}] (13rect) {};

    \draw[->, darkgreen, bend left = 30] (7rect) to node[above] {$c$} (13rect.west);
    
    \draw[->, blue, bend left = 45] ($(2rect.west)+(0,-0.25)$) to node[below] {$b$} ($(7rect.east)+(0,-0.25)$);

    \draw[->, darkgreen] (13rect.north) to [bend right = 15] ($(13rect.north)+(-1,0.3)$) to node[below] {$c$} ($(15rect.north)+(-0.5*\xdistrect,0.3)$) to [bend right = 30] ($(14rect.north)+(-0.15,0)$);
    \end{scope}
\end{tikzpicture}
}
\vspace{0.5em}

\setlength{\tabcolsep}{0.2em} 
\renewcommand{\arraystretch}{1.1}
\resizebox{0.95\textwidth}{!}{
\begin{tabular}{c c c c c c c c c c c c c c c c c c c c c c c c c c c c c c c c}

\toprule
\multicolumn{2}{c}{1} & \multicolumn{2}{c}{2} & \multicolumn{2}{c}{3} & \multicolumn{2}{c}{4} & \multicolumn{2}{c}{5} & \multicolumn{2}{c}{6} & \multicolumn{2}{c}{7} & \multicolumn{2}{c}{8} & \multicolumn{2}{c}{9} & \multicolumn{2}{c}{10} & \multicolumn{2}{c}{11} & \multicolumn{2}{c}{12} & \multicolumn{2}{c}{13} & \multicolumn{2}{c}{14} & \multicolumn{2}{c}{15} & \multicolumn{2}{c}{16}  \\ [-2px]

\tiny{$\inf$} & \tiny{$\sup$} & \tiny{$\inf$} & \tiny{$\sup$} & \tiny{$\inf$} & \tiny{$\sup$} & \tiny{$\inf$} & \tiny{$\sup$} & \tiny{$\inf$} & \tiny{$\sup$} & \tiny{$\inf$} & \tiny{$\sup$} & \tiny{$\inf$} & \tiny{$\sup$} & \tiny{$\inf$} & \tiny{$\sup$} & \tiny{$\inf$} & \tiny{$\sup$} & \tiny{$\inf$} & \tiny{$\sup$} & \tiny{$\inf$} & \tiny{$\sup$} & \tiny{$\inf$} & \tiny{$\sup$} & \tiny{$\inf$} & \tiny{$\sup$} & \tiny{$\inf$} & \tiny{$\sup$} & \tiny{$\inf$} & \tiny{$\sup$} & \tiny{$\inf$} & \tiny{$\sup$} \\ [1px]

\hline & \\[-10px]

\small $\#$ & \small $\#$ & a & a & a & a & a & a & a & a & a & a & b & b & b & b & b & b & b & b & c & c & c & c & c & c & c & c & c & c & c & c \\ [-4px]

\small $\#$ & \small $\#$ & \small $\#$ & a & b & b & b & b & c & c & c & c & a & a & a & c & c & c & c & c & a & a & a & a & a & a & b & b & c & c & c & c \\ [-4px]

\small $\#$ & \small $\#$ & \small $\#$ & c & a & c & c & c & b & b & c & c & b & b & c & a & c & c & c & c & \small $\#$ & a & b & b & c & c & a & c & a & a & c & c \\ [-4px]

\small $\#$ & \small $\#$ & \small $\#$ & c & b & a & c & c & a & c & a & a & a & c & c & b & a & a & c & c & \small $\#$ & c & c & c & b & b & c & a & \small $\#$ & c & a & a \\ [-4px]

\small $\#$ & \small $\#$ & \small $\#$ & a & a & b & a & a & c & a & \small $\#$ & c & b & a & a & c & \small $\#$ & c & a & a & \small $\#$ & c & c & c & a & c & c & b & \small $\#$ & b & \small $\#$ & c \\ [-4px]

\small $\#$ & \small $\#$ & \small $\#$ & c & b & c & \small $\#$ & c & c & b & \small $\#$ & b & a & b & \small $\#$ & c & \small $\#$ & b & \small $\#$ & c & \small $\#$ & a & a & a & c & a & a & c & \small $\#$ & c & \small $\#$ & b \\ [-4px]

\small $\#$ & \small $\#$ & \small $\#$ & b & a & c & \small $\#$ & b & a & c & \small $\#$ & c & b & c & \small $\#$ & a & \small $\#$ & c & \small $\#$ & b & \small $\#$ & c & \small $\#$ & c & c & b & \small $\#$ & c & \small $\#$ & a & \small $\#$ & c \\ [-4px]

\small $\#$ & \small $\#$ & \small $\#$ & c & b & a & \small $\#$ & c & \small $\#$ & c & \small $\#$ & a & a & c & \small $\#$ & c & \small $\#$ & a & \small $\#$ & c & \small $\#$ & b & \small $\#$ & b & a & c & \small $\#$ & a & \small $\#$ & b & \small $\#$ & a \\

\bottomrule
\end{tabular}
} 
\end{center}

\caption{The figure shows: $(a)$ a DFA $\mathcal D$, $(b)$ the sequences $\mathsf{OUT\_DEG}$ and $\mathsf{OUT}$ representing the outgoing transitions of $\mathcal D$, $(c)$ the infimum pseudoforest of $\mathcal D$, $(d)$ the supremum pseudoforest of $\mathcal D$, and $(e)$ the DFA $\mathcal D$ partitioned into its maximal unary paths.
The table below shows the infimum and supremum strings in $\mathcal D$.}
\label{fig: example}
\end{figure}
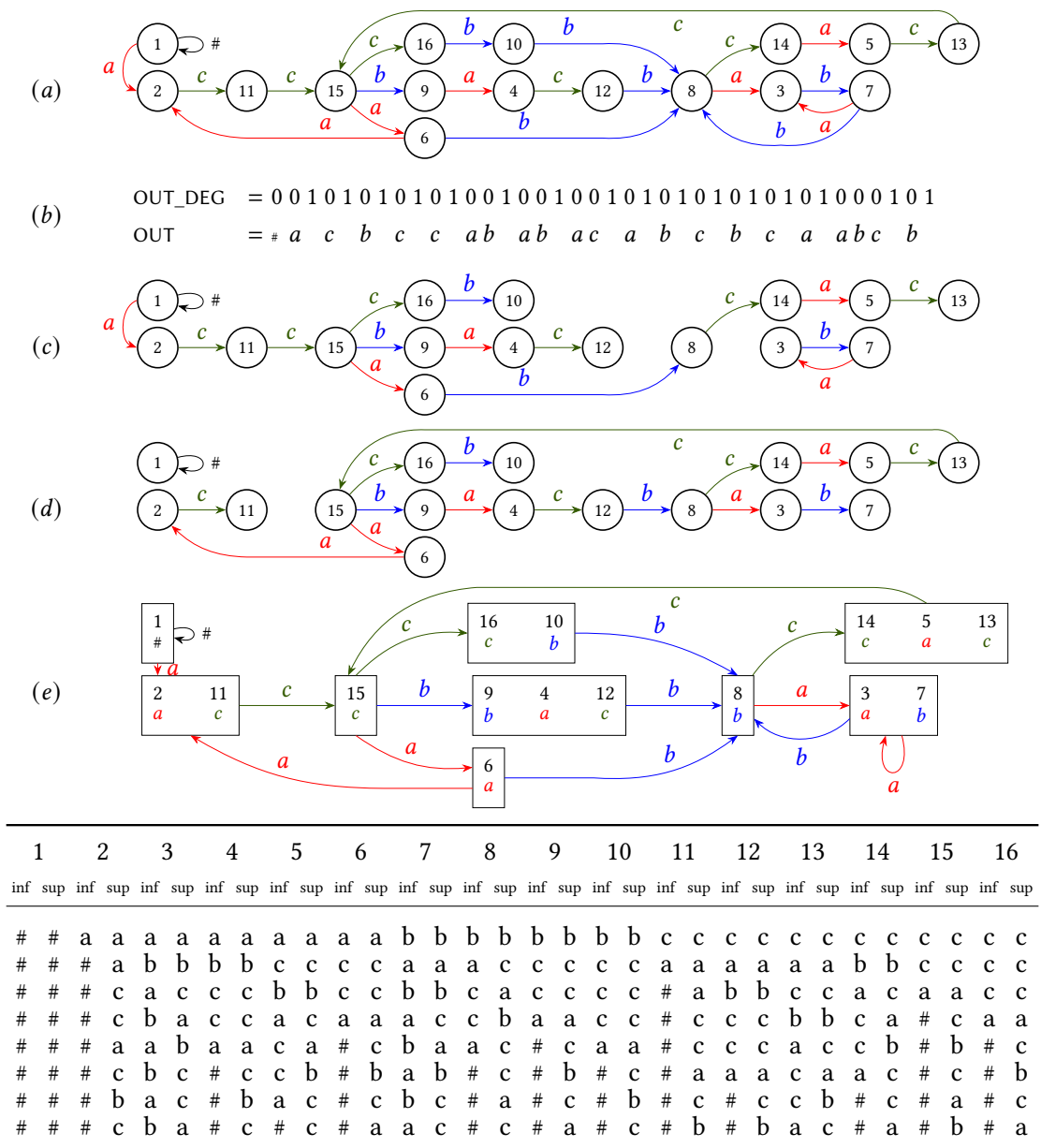

\section{Experiments}\label{sec: experiments}

We implemented our novel indexing technique based on the Graph Suffix Array (Algorithm~\ref{alg:computeT}), in C++ and made the source code available at \url{https://github.com/regindex/Gr-SA/tree/ictcs-paper}.
To assess its performance, we compared it with the classical forward search algorithm for BWT-based indexes.
Since to the best of our knowledge there exists no public code of this latter algorithm, we also provided an implementation of the forward search on C++.
We ran our algorithm on Wheeler pangenome graphs constructed from the dbSNP database, encoding the Finnish subset of the frequent DNA mutations~\cite{HelsinkiData}.
Specifically, we use pangenome graphs constructed from the genomic data related to the chromosomes 21 and 14.
All experiments were run on a server equipped with an Intel Xeon W-2245 CPU (3.90 GHz, 8 cores) and 128 GB of RAM, running Ubuntu 18.04 LTS 64-bit.

\paragraph{Implementation details.}
We emphasize that our practical implementation differs from the description provided in Section~\ref{sec:data-strucure} in the following aspects.
(i) Our solution is optimized for alphabets $\Sigma$ of size $4$, i.e., $\Sigma = \{A, C, G, T \}$.
As consequence of that, we do not use minimal perfect hash functions to navigate between unary paths, but we store the outgoing transitions of the last state $v$ in the unary path through contiguous pairs $(a, u)$, where $a$ is the label of the outgoing transition, and $u$ is a pointer to the destination state.
Scanning this sequence of pairs $(a, u)$ will be particularly fast since $\lvert out(v) \rvert \leq 4$.
(ii) To save memory space, we do not store explicitly the infimum and supremum pseudoforests.
Instead, we use the information contained in the plain representation of the DFA to read the infimum and supremum string when we call the oracle $\oracle$.
In particular, the graph is linearized in such a way to form a HLD on the infimum pseudoforest.
Since most of the nodes have only one incoming edge, this linearization will decrease the number of cache misses also when we read supremum strings.
(iii) We store in a table $T$ all the intervals $]i,j[$ returned by the oracle $\oracle(\beta')$ for every string $\beta'$ of length $k$. This parameter $k$ is chosen is such a way that the table $T$ occupies at most $40\%$ of the total space used by the other components of the index.
When we call $\oracle(\beta)$, we fetch the integers $i$ and $j$ from $T[\beta']$, where $\beta'$ is the suffix of $\beta$ of length $k$.
This allows us to run the binary search of Line~\ref{alg:computeT function locate states} on the restricted interval $[i,j]$, instead of $[1,2n]$.
All these heuristics result in a significant improvement of the time and space performance of our algorithm.

\begin{figure}[htbp]
    \centering
    \includegraphics[width=0.8\textwidth]{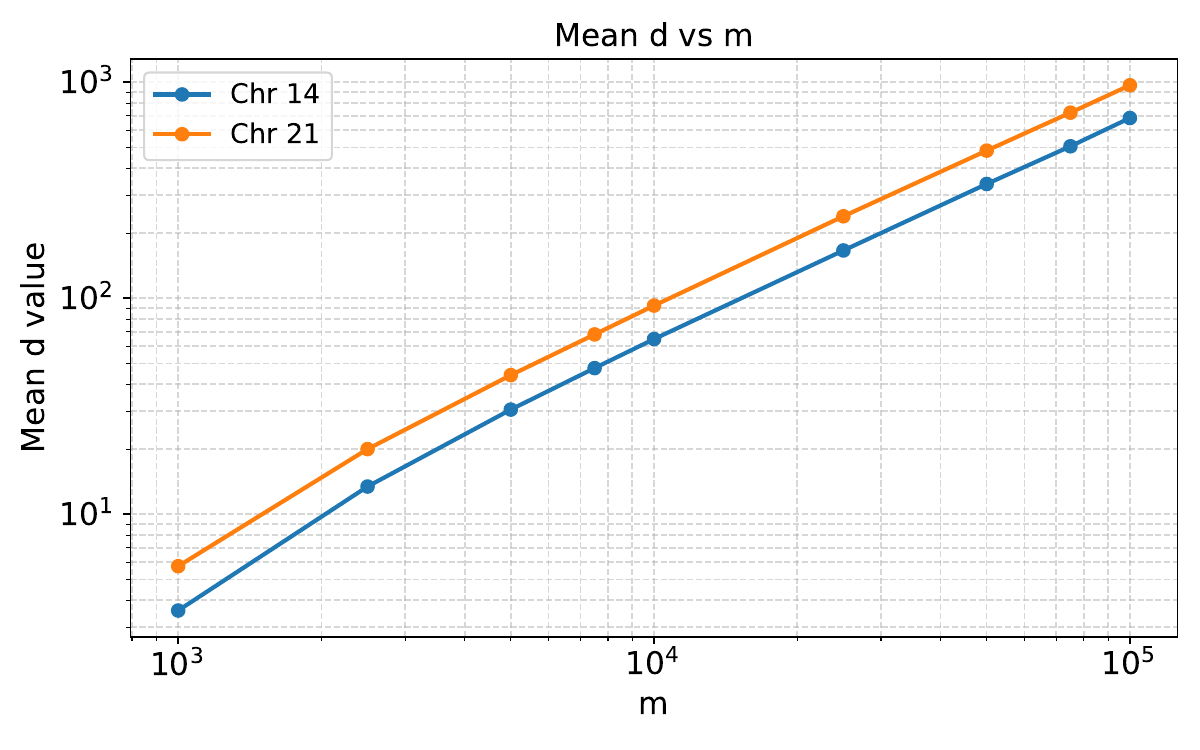}
    \caption{Mean $d$ (number of traversed unary paths) measured during locate($\alpha$) versus length $m$ (log-log scale). The mean is computed across $10^4$ random queries, where each string $\alpha$ of length $m$ is sampled from the graph.}
    \label{fig:d-plot}
\end{figure}

\paragraph{Experiments description.}
We conduct two types of experiments.
In the first experiment we empirically measure the value $d$, representing the number of encountered unary paths in the navigation step.
To do so, we perform $10^4$ locate queries for every analyzed length of the pattern, and we measure the number of unary path traversed in the third phase of the algorithm.
The results are shown in Figure~\ref{fig:d-plot}.
As we can observe, due to their sparseness, the typical values of $d$ are on both graphs about two orders of magnitude smaller compared to the pattern length.
This shows that the number of cache misses triggered during the third phase of Algorithm~\ref{alg:computeT} (i.e., from Line~\ref{alg:computeT navigating start} to Line~\ref{alg:computeT navigating end}) does not have a significant impact on the algorithm performance.
Next, we compare the time and space efficiency of our index (Gr-SA) with respect to the original forward search (Fw.) algorithm for Wheeler graphs (Fw.)~\cite{Wheeler-graphs}.
We considered only patterns appearing in the graphs and we reported the average time to process a single character of those patterns.
The results are summarized in Table~\ref{tab:dataset_properties}.

\begin{table}[htbp]
\centering
\resizebox{\textwidth}{!}{
\begin{tabular}{lcccccccccc}
\toprule
\multirow{2}{*}{\textbf{Dataset}} & \multicolumn{2}{c}{\textbf{WDFA Properties}} & \multicolumn{3}{c}{\textbf{Index Size}} & \multicolumn{2}{c}{\textbf{Queries}} & \multicolumn{3}{c}{\textbf{Execution Time}} \\
\cmidrule(r){2-3} \cmidrule(lr){4-6} \cmidrule(lr){7-8} \cmidrule(l){9-11}
& $n$ & $\vert \delta\rvert$ & Fw. & Gr-SA & Blowup & $m$ & $m/\lvert\beta\rvert$ & Fw. & Gr-SA & Speedup \\
\midrule

\multirow{2}{*}{Chr $21$} & \multirow{2}{*}{$3.85\times 10^7$} & \multirow{2}{*}{$3.88\times 10^7$} & \multirow{2}{*}{$38.1$ MB} & \multirow{2}{*}{$528.9$ MB} & \multirow{2}{*}{$\times13.9$} & $10^3$ & $0.58$ & $968.5$ ns & $9.4$ ns & $\times103$\\
& & & & & & $10^4$ & $0.09$ & $976.8$ ns & $2.4$ ns & $\times410$ \\
\midrule
\multirow{2}{*}{Chr $14$} & \multirow{2}{*}{$9.44\times 10^7$} & \multirow{2}{*}{$9.47\times 10^7$} & \multirow{2}{*}{$93.2$ MB} & \multirow{2}{*}{$1374.7$ MB} & \multirow{2}{*}{$\times14.8$} & $10^3$ & $0.70$ & $1199.0$ ns & $10.0$ ns & $\times119$\\
& & & & & & $10^4$ & $0.30$ & $1192.8$ ns & $2.4$ ns & $\times501$\\
\bottomrule
\end{tabular}
} 
\caption{The table shows the space usage and the average query time used by the forward search (Fw) and the Graph Suffix Array (Gr-SA) on the pangenomes built from chromosomes 21 and 14.
Specifically, the table shows: (i)~the number of states and transitions in the pangenome graph; (ii)~the total space usage for both indexes; (iii)~the length $m$ of the patterns and the average ratio of $m$ to the length of the longest prefix appearing at least twice in the graph; (iv)~the average execution time per character of the pattern.}
\label{tab:dataset_properties}
\end{table}

\end{document}

%% file: arxiv_macros.tex
\let\epsilon\varepsilon

\usepackage{xspace}

\usepackage[linesnumbered,ruled,vlined,boxed]{algorithm2e}
\DontPrintSemicolon
\SetKwInOut{Input}{Input}\SetKwInOut{Output}{Output}
\SetKwProg{Fn}{Function}{:}{}
\SetKwFunction{LocateStates}{LocateStates}

\usepackage[dvipsnames]{xcolor}

\usepackage{tikz}
\usepackage{tikz}
\usetikzlibrary{automata}
\usetikzlibrary{positioning}
\usetikzlibrary{math} 
\usetikzlibrary{arrows.meta}
\usetikzlibrary{fit}
\usetikzlibrary{calc}

\PassOptionsToPackage{usenames,dvipsnames,svgnames}{xcolor}
\definecolor{orange}{RGB}{235,90,0}
\definecolor{darkorange}{RGB}{175,30,0}
\definecolor{turkis}{RGB}{131,182,182}
\definecolor{darkturkis}{RGB}{31,82,82}
\definecolor{green}{RGB}{102,180,0}
\definecolor{darkgreen}{RGB}{51,90,0}
\definecolor{myblue}{RGB}{0,0,213}
\definecolor{mydarkblue}{RGB}{0,0,100}
\definecolor{mybrightblue}{HTML}{74B0E4}
\definecolor{mybrighterblue}{HTML}{B3EAFA}
\definecolor{lila}{RGB}{102,0,102}
\definecolor{darkred}{RGB}{139,0,0}
\definecolor{darkyellow}{RGB}{188,135,2}
\definecolor{brightgray}{RGB}{200,200,200}
\definecolor{darkgray}{RGB}{50,50,50}
\definecolor{amaranth}{rgb}{0.9, 0.17, 0.31}
\definecolor{alizarin}{rgb}{0.82, 0.1, 0.26}
\definecolor{amber}{rgb}{1.0, 0.75, 0.0}
\definecolor{green(ryb)}{rgb}{0.4, 0.69, 0.2}
\definecolor{hanblue}{rgb}{0.27, 0.42, 0.81}
\definecolor{grannysmithapple}{rgb}{0.66, 0.89, 0.63}

\newtheorem{theorem}{Theorem}[section]
\newtheorem{lemma}[theorem]{Lemma}
\newtheorem{definition}[theorem]{Definition}
\newtheorem{corollary}[theorem]{Corollary}

\newtheorem{remark}[theorem]{Remark}

\newenvironment{acknowledgements}{%
  \begin{abstract}
}{%
  \end{abstract}
}